\documentclass[11p,reqno]{amsart}

\usepackage[T1]{fontenc}
\usepackage{graphicx}
\usepackage{amssymb,amsthm,amsmath,mathrsfs,bm,braket,marginnote}
\usepackage{enumerate}
\usepackage{appendix}
\usepackage[colorlinks=true, pdfstartview=FitV, linkcolor=blue, citecolor=blue, urlcolor=blue]{hyperref}
\usepackage{multirow}

\newcommand{\MeijerG}[8][\bigg]{G^{{ #2 },{ #3 }}_{{ #4 },{ #5 }} #1( \begin{matrix} #6 \\ #7 \end{matrix}\, #1\vert\, #8 #1)}

\usepackage{pgf}
\usepackage{pgfplots}
\usepackage{tikz}
\usetikzlibrary{arrows,calc}
\usepackage{verbatim}
\usetikzlibrary{decorations.pathreplacing,decorations.pathmorphing}
\usepackage[numbers,sort&compress]{natbib}
\usepackage{dsfont}

\numberwithin{equation}{section}
\newtheorem{theorem}{Theorem}[section]

\newtheorem{proposition}[theorem]{Proposition}

\theoremstyle{remark}
\newtheorem{remark}[theorem]{Remark}

 \reversemarginpar
\newcommand{\ta}{\tilde{a}}
\newcommand{\tb}{\tilde{b}}

\begin{document}

\title[Rate of convergence for hard edge scaling]{Rate of convergence at the hard edge for various P\'olya ensembles of
positive definite matrices}

\author{Peter J. Forrester}
\address{School of Mathematical and Statistics, ARC Centre of Excellence for Mathematical and Statistical Frontiers, The University of Melbourne, Victoria 3010, Australia}
\email{pjforr@unimelb.edu.au}

\author{Shi-Hao Li}
\address{ School of Mathematical and Statistics, ARC Centre of Excellence for Mathematical and Statistical Frontiers, The University of Melbourne, Victoria 3010, Australia}
\email{lishihao@lsec.cc.ac.cn}

\date{}

\dedicatory{}

\keywords{}

\begin{abstract}
The theory of P\'olya ensembles of positive definite random matrices provides structural formulas
for the corresponding biorthogonal pair, and correlation kernel, which are well suited to computing the hard edge large
$N$ asymptotics. Such an analysis is carried out for
 products of Laguerre ensembles, the Laguerre Muttalib-Borodin ensemble, and products of Laguerre ensembles and their inverses.
 The latter includes as a special case the Jacobi unitary ensemble. In each case the hard edge scaled kernel permits an
 expansion in powers of $1/N$, with the leading  term given   in a structured form involving the hard edge scaling of the
 biorthogonal pair. The Laguerre and Jacobi ensembles have the special feature that their hard edge scaled kernel ---
 the Bessel kernel --- is symmetric and this leads to there being a choice of hard edge scaling variables for which
 the rate of convergence of the correlation functions is $O(1/N^2)$.
\end{abstract}

\maketitle

\section{Introduction}
There are many settings in random matrix theory for which the eigenvalues (assumed real) can be scaled in relation to the
matrix size in such a way that the limiting support is compact. This is referred to as a global scaling. As some concrete examples,
let $X$ be an $N \times N$ standard complex Gaussian matrix, and construct from this the Hermitian matrices
$H_1 = {1 \over 2} (X + X^\dagger)$ and $H_2 = X^\dagger X$. The set of matrices $H_1$ ($H_2$) are said to form the
Gaussian unitary ensemble (special case of the Laguerre unitary ensemble), and have joint eigenvalue probability
density function (PDF) proportional to
\begin{equation}\label{1.1}
\prod_{l=1}^N w(x_l) \prod_{1 \le j < k \le N} (x_k - x_j)^2, \qquad w(x) = \left \{ \begin{array}{cc} e^{-x^2}, & {\rm matrices} \: H_1 \\
 e^{-x} \chi_{x > 0}, & {\rm matrices} \: H_2; \end{array} \right.
 \end{equation}
see e.g.~\cite{Fo10,PS11}. Here $\chi_A = 1$ for $A$ true, $\chi_A = 0$ otherwise.

Scaling the eigenvalues $x_j \mapsto \sqrt{2N} x_j $ (matrices $H_1$) and $x_j \mapsto 4 N x_j $ (matrices $H_2$), it is a standard
result that as $N \to \infty$ the spectrum is supported on the intervals $(-1,1)$ and $(0,1)$ respectively. Among the endpoints of
the intervals of support, the point $x = 0$ for the global scaling of the matrices $H_2$ is special. Thus the region $x < 0$
to the other side of this endpoint has strictly zero eigenvalue density for all values of $N$, because $H_2$ is positive
definite. For this reason the endpoint $x = 0$ in this example is called a hard edge. The hard edge notion extends beyond the class of
matrix ensembles permitting a global scaling to include heavy tailed distributions --- an example of the latter is given in
Section \ref{S3.3} below. The essential point then is that the limiting eigenvalue density is nonzero for $x>0$, and strictly
zero for $x<0$.

In this paper our interest is in the approach to a limiting hard edge state for various ensembles of positive definite matrices.
A hard edge state refers to the statistical distribution formed when the eigenvalues are scaled to have nearest neighbour
spacing of order unity as $N \to \infty$. For the matrices $H_2$, or more generally the ensemble of matrices with
weight function 
\begin{equation}\label{Lw}
w(x) = x^a e^{-x} \chi_{x > 0}
 \end{equation}
 (Laguerre weight, realised for $a = n - N \in \mathbb Z_{\ge 0}$ as the eigenvalue PDF of
 matrices $X^\dagger X$ with $X$ an $n \times N$ complex standard Gaussian matrix) with parameter $a > - 1$,  this takes place for the scaling of the eigenvalues
$x_j \mapsto x_j / 4N$, and gives rise to the hard edge state specified by the $k$-point correlations (see \cite[\S 7.2]{Fo10})
\begin{equation}\label{1.2}
\rho_{(k)}^{\rm hard}(x_1,\dots,x_k) = \det [ K^{\rm hard}(x_j, x_l;a) ]_{j,l=1}^k,
 \end{equation}
 where, with $J_a(u)$ denoting the Bessel function,
\begin{equation}\label{1.3}
K^{\rm hard}(x,y;a)  
 = {1 \over 4} \int_0^1 J_a(\sqrt{xt})  J_a(\sqrt{yt}) \, dt.
\end{equation}

For finite $N$ the $k$-point correlation function is defined in terms of the joint eigenvalue PDF, $P_N$ say,
according to
\begin{equation}\label{1.4}
\rho_{(k)}(x_1,\dots,x_k) = {N! \over (N - k)!} \int_{-\infty}^\infty dx_{k+1} \cdots  \int_{-\infty}^\infty dx_N \, P_N(x_1,\dots, x_N).
 \end{equation}
 For eigenvalue PDFs of the form (\ref{1.1}), the correlation function (\ref{1.4}) admits the determinant evaluation (see e.g.~\cite[\S 5.1]{Fo10})
\begin{equation}\label{1.5a} 
\rho_{(k)}(x_1,\dots,x_k) =  \det [ K_N(x_j, x_l ) ]_{j,l=1}^k,
 \end{equation}
 where
 \begin{align}\label{1.5b} 
  K_N(x,y)  = \Big ( w(x) w(y) \Big )^{1/2} \sum_{n=0}^{N-1} {1 \over h_n} p_n(x) p_n(y) \nonumber \\
 \end{align}
 In (\ref{1.5b}) $\{p_n(x)\}$ refers to the set of orthogonal polynomials with respect to the weight function $w(x)$ ---
 $p_n$ of degree $n$ and chosen to be monic for convenience --- with norm $h_n$,
 \begin{equation}\label{1.5c}   
 \int_{-\infty}^\infty w(x) p_m(x) p_n(x) \, dx = h_n \delta_{m,n}.
  \end{equation}
  In the case of the Laguerre weight, the polynomials $p_n(x)$ are proportional to the Laguerre
  polynomials $L_n^{(a)}(x)$.
  
  Recently, attention has been given to  the rate of convergence to the hard edge limiting kernel
  (\ref{1.3}). One line of motivation came from a question posed by
 Edelman, Guionnet and P\'ech\'e \cite{EGP16}. These authors, taking a viewpoint in numerical analysis, took up the problem of studying
 finite $N$ effects in the hard edge scaling of the distribution of the smallest singular value of a (complex) standard Gaussian matrix. 
 With $E^{\rm LUE}(0;(0,s))$ denoting the probability that there are no eigenvalues in the interval $(0,s)$ of the LUE, it was conjectured
 in \cite{EGP16} that
 \begin{equation}\label{pk5} 
E^{\rm LUE}(0;(0,s/(4N))) =  E^{\rm hard}(0;(0,s)) + {a \over 2N} s {d \over ds} E^{\rm hard}(0;(0,s)) + O \Big ( {1 \over N^2} \Big ),
 \end{equation}
 where
 $$
 E^{\rm hard}(0;(0,s)) = \lim_{N \to \infty}  E^{\rm LUE}(0;(0,s/(4N))),
 $$
 and thus  \cite{Bo16,PS16},
  \begin{equation}\label{pk7} 
 E^{\rm LUE}\bigg ( 0; \Big ( 0, {s \over 4 N + 2a} \Big )  \bigg ) =  E^{\rm hard}(0;(0,s)) + O \Big ( {1 \over N^2} \Big ),
 \end{equation} 
which moreover is the optimal rate of convergence. 

 Subsequently Bornemann \cite{Bo16}
 provided a proof of (\ref{pk5}) which involved extending the limit formula (\ref{1.2}) to the large $N$ expansion
 \begin{align}\label{pk6} 
 {1 \over 4 N} K_N^{(L)} \Big ( {X \over 4 N} , {Y \over 4 N} \Big ) & =  K^{\rm hard}(X,Y)  + {1 \over N} {a \over 8} J_a(\sqrt{X})  J_a(\sqrt{Y})  + O \Big ( {1 \over N^2} \Big ) \nonumber \\
 &  =  K^{\rm hard}(X,Y)  + {1 \over N}  {a \over 2} \bigg ( x {\partial \over \partial x} +   y {\partial \over \partial y}  + 1 \bigg )  K^{\rm hard}(X,Y) + O \Big ( {1 \over N^2} \Big ),
 \end{align}
 valid uniformly for $X,Y \in [0,s]$. In fact knowledge of (\ref{pk6}) is sufficient to establish (\ref{pk5}).
 We remark too that analogous to (\ref{pk7}), it follows from  (\ref{pk6}) that
\begin{equation}\label{pk8} 
 {1 \over 4 N + 2a } K_N^{(L)} \Big ( {X \over 4 N + 2a} , {Y \over 4 N + 2a} \Big )  =  K^{\rm hard}(X,Y)  + O \Big ( {1 \over N^2} \Big ),
 \end{equation}
 and this implies (\ref{pk7}).
 
 Our aim in this work is to extend hard edge scaling results of the type (\ref{pk6}) to examples of a recently isolated structured class of random
 matrices known as P\'olya ensembles \cite{KK16}. The definition of these ensembles, which include the Laguerre unitary ensemble, the Jacobi unitary ensemble,
 products of these ensembles, and their Muttalib-Borodin generalisations, will be given in Section \ref{S2.1}. The benefit of the structures provided
 by the P\'olya ensemble class is seen by our revision of the key formulas in Section \ref{S2.2}, where we also extend the theory by exhibiting differential
 recurrences satisfied by the associated biothogonal pair, and a differential identity satisfied by the correlation kernel. In Section \ref{S2.3} we make note of some
 asymptotic formulas relating to ratios of gamma functions which will be used in our subsequent large $N$ hard edge analysis. The latter is undertaken is Section \ref{S3},
 starting with products of Laguerre ensembles, then the Laguerre Muttalib-Borodin ensemble, and finally products of Laguerre ensembles and their inverses,
 with the latter including as a special case the Jacobi unitary ensemble.
 
 The Jacobi unitary ensemble is specified by the eigenvalue PDF (\ref{1.1}) with weight 
 \begin{equation}\label{A.2}
 x^a (1 - x)^b \chi_{0 < x < 1}.
 \end{equation}
 Our results of Section \ref{S3.3} imply that
  \begin{align}\label{pk6a} 
 {1 \over 4 N^2} K_N^{(J)} \Big ( {X \over 4 N^2} , {Y \over 4 N^2} \Big ) & =  K^{\rm hard}(X,Y)  +  { a+b \over 2 N} J_a(\sqrt{X})  J_a(\sqrt{Y})  + O \Big ( {1 \over N^2} \Big ) \nonumber \\
 &  =  K^{\rm hard}(X,Y)  +   { a +b \over N}  \bigg ( x {\partial \over \partial x} +   y {\partial \over \partial y}  + 1 \bigg )  K^{\rm hard}(X,Y) + O \Big ( {1 \over N^2} \Big ),
 \end{align}
and thus 
 \begin{equation}\label{A.2a}
  {1 \over 4 \tilde{N}^2} K_N^{(J)} \Big ( {X \over 4 \tilde{N}^2} , {Y \over 4 \tilde{N}^2} \Big )  \bigg |_{\tilde{N} = N + (a+b)/2} =
   K^{\rm hard}(X,Y)  +   O \Big ( {1 \over N^2} \Big ).
   \end{equation}
   This gives an explanation for recent results in \cite{MMM19} relating to the large $N$ form of the distribution
   of the smallest eigenvalue in the Jacobi unitary ensemble. In Appendix A large $N$ expansions of the latter
   quantity are extended to all Jacobi $\beta$-ensembles with $\beta$ even.

  \section{Preliminaries}
  \subsection{P\'olya ensembles --- definitions}\label{S2.1}
  The Vandermonde determinant identity tells us that
 \begin{equation}\label{2.1}
 \det [ x_k^{j-1} ]_{j,k=1}^N = \det [p_{j-1}(x_k) ]_{j,k=1}^{N} = \prod_{1 \le j < k \le N} (x_k - x_j),
   \end{equation}
   where $\{p_{l}(x)\}_{l=0}^{N-1}$ are arbitrary monic orthogonal polynomials, $p_l$ of degree $l$.
   A generalisation of (\ref{1.1}) is therefore an eigenvalue PDF proportional to
  \begin{equation}\label{2.2}
   \det [p_{j-1}(x_k) ]_{j,k=1}^{N} \det [ w_{j-1}(x_k) ]_{j,k=1}^N
 \end{equation}
 for some polynomials $\{p_{l}(x)\}_{l=0}^{N-1}$ and functions
 $\{ w_j(x) \}_{j=0}^{N-1}$ --- note though that in general there is no guarantee
   (\ref{2.2}) will be positive. In \cite{KZ14} eigenvalue PDFs (\ref{2.2}) were given the name
 polynomial ensembles.    
 
 In \cite{KK16,KK19} a further specialisation of (\ref{2.2}),
  \begin{equation}\label{2.3} 
\det [p_{j-1}(x_k) ]_{j,k=1}^{N}   \det \bigg [ \Big ( - x_k {\partial \over \partial x_k} \Big )^{j-1} w(x_k) \bigg ]_{j,k=1}^N,
 \end{equation}
 was proposed. Assuming all the eigenvalues are positive, it was  shown that this class of eigenvalue
 PDF is closed under multiplicative convolution. At first  PDFs of the form (\ref{2.3}) were referred to as
 polynomial ensembles of derivative type, but subsequently with the requirement that they be non-negative,
 it was pointed out in \cite{FKK17} that it is more apt to use the term P\'olya ensemble.  The invariance of a
 determinant under the elementary row operation of adding one multiple of a row to another shows
   \begin{align}\label{2.4} 
 \det \bigg [ \Big ( - x_k {\partial \over \partial x_k} \Big )^{j-1} w(x_k) \bigg ]_{j,k=1}^N
 & =     \det \bigg [  \prod_{l=1}^{j-1} \Big ( - x_k {\partial \over \partial x_k} - l \Big ) w(x_k)  \bigg ]_{j,k=1}^N \nonumber \\
 & =  \det \bigg [ {\partial^{j-1} \over \partial x_k^{j-1}} \Big ( (-x_k)^{j-1} w(x_k) \Big ) \bigg ]_{j,k=1}^N.
 \end{align}
 In relation to the second line, note that it is in fact an equality that
   \begin{equation}\label{2.5} 
    \prod_{l=1}^{j-1} \Big ( - x {\partial \over \partial x} - l \Big ) w(x) = {d^{j-1} \over d x^{j-1}} \Big ( (-x)^{j-1} w(x) \Big ).
 \end{equation}    
 
 The differential operator on the RHS of (\ref{2.5}) reveals that the Laguerre unitary ensemble fits the framework of
P\'olya ensembles. Thus choosing $w(x)$ to be given by (\ref{Lw}), the Rodrigues formula for the Laguerre polynomials
tells us that
  \begin{equation}\label{2.6} 
  {d^{j-1} \over d x^{j-1}} \Big ( (-x)^{j-1} w(x) \Big ) = (-1)^{j-1} (j-1)! w(x) L_{j-1}^{(a)}(x),
  \end{equation} 
  and so, up to proportionality, (\ref{2.3}) reduces to
   \begin{equation}\label{2.7}   
   \prod_{l=1}^N x_l^a e^{-x_l} \det [ p_{j-1}(x_k) ]_{j,k=1}^N    \det [ L_{j-1}^{(a)}(x_k) ]_{j,k=1}^N.
   \end{equation} 
  In view of (\ref{2.1}), this corresponds to the eigenvalue PDF for the Laguerre unitary ensemble. The advantage in  working
  within the P\'olya ensemble framework is that it reveals a mechanism to obtain the  asymptotic expansion of the
  correlation kernel  (\ref{1.5b}) at the hard edge, which applies at once to a much wider class of random matrix ensembles.
  The reason for this are certain general structural formulas applicable to all P\'olya ensembles. These will be revised next.
  
  \subsection{P\'olya ensembles --- biorthogonal system and correlation kernel}\label{S2.2}
  It is standard in random matrix theory that the ensembles (\ref{2.2}) are determinantal, meaning that the $k$-point correlation
  functions have the form (\ref{1.5a}). Moreover, if the polynomials $\{p_l(x) \}_{l=0}^N$ and the functions $\{q_j(x)\}_{j=0}^N$ --- the latter
  chosen from ${\rm span} \, \{w_j(x) \}_{j=0}^{N}$ --- have the
  biorthogonal property
     \begin{equation}\label{2.8} 
 \int_{-\infty}^\infty p_m(x) q_n(x) \, dx = \delta_{m,n},
 \end{equation}   
 then the correlation kernel has the simple form
    \begin{equation}\label{2.9}   
    K_N(x,y) = \sum_{j=0}^{N-1} p_{j}(x) q_{j}(y);
 \end{equation}   
 see e.g.~\cite[\S 5.8]{Fo10}.
 While in general computation of  the LU (lower/ upper triangular) decomposition of a certain inverse matrix
  used to construct the biorthogonal
 functions (see e.g.~\cite[Proof of Prop.~5.8.1]{Fo10}), this cannot be expected to result in
a tractable    formula for (\ref{2.9}), permitting large $N$ analysis, without further structures.
It is at this stage that the utility of P\'olya ensembles shows itself: special functional forms for the
 biorthogonal system hold true, and moreover there is a summed up form of the kernel
 as an integral analogous to (\ref{1.3}), which together facilitate a large $N$ analysis.
 
 The formulas, which are due to Kieburg and K\"osters \cite{KK16}, involve the Mellin transform
 of the weight $w$ in (\ref{2.3}),
   \begin{equation}\label{Me}
   \mathcal M[w](s) := \int_0^\infty y^{s-1} w(y) \, dy.
  \end{equation}  
One has that the polynomials $\{p_l(x)\}_{l=0}^N$ in the biorthogonal pair $\{p_j, q_k\}$ are specified by
   \begin{equation}\label{Me1} 
 p_n(x) = (-1)^n  n!   \mathcal M[w](n+1) \sum_{j=0}^n{  (-x)^{j}  \over j! (n-j)!   \mathcal M[w](j+1) } ,
 \end{equation}  
and that the functions $\{q_l(x)\}_{l=0}^N$  --- chosen from the span of the functions specifying the columns in (\ref{2.3}) ---
are specified by the Rodrigues type formula
  \begin{equation}\label{Me2} 
  q_n(x) = {1 \over n! \mathcal M[w](n+1)} {d^n \over d x^n} \Big((-x)^n w(x)\Big).
 \end{equation} 
 Moreover, the correlation kernel can be written in a form generalising the final expression in (\ref{1.3}),
  \begin{equation}\label{Me3}    
  K_N(x,y) = - N { \mathcal M[w](N+1) \over  \mathcal M[w](N)}  \int_0^1 p_{N-1}(xt) q_N(yt) \, dt.
   \end{equation} 
  
  In \cite{KK16} the integral form (\ref{Me3}) of the correlation kernel was derived by first converting
  (\ref{Me1}) and (\ref{Me2}) to integral forms, which allow for the summation to be carried out
  in closed form. The identification with the RHS of (\ref{Me3}) then follows after some manipulation.
  In a special case this strategy was first given in \cite{KZ14}.
  An alternative method of derivation is also possible, as we will now show, which involves first
  identifying differential recurrences satisfied by each of the $p_n(x)$ and $q_n(x)$.
  (We remark that other examples of differential recurrences can be found in a number of recent studies in random matrix
  theory \cite{Ku19,FK19,FK20,Fo20a,Fo20b}.)
  
  \begin{proposition}\label{PR1}
  Let $p_n(x)$ and $q_n(x)$ be specified by (\ref{Me1}) and (\ref{Me2}).
  These functions satisfy the differential recurrences
  \begin{align}
  x {d \over dx} p_n(x) & =  n p_n(x) + n {\mathcal M [w](n+1) \over \mathcal M [w] (n) } p_{n-1}(x) \label{PQ1} \\
  x {d \over dx} q_n(x) & =  - {(n+1) \mathcal M[w](n+2) \over  \mathcal M[w](n+1)} q_{n+1}(x) + (n+1) q_n(x).\label{PQ2}
  \end{align}
  A corollary of these recurrences is the differential identity
   \begin{equation}\label{Me4} 
  \Big ( x{\partial \over \partial x} + y{\partial \over \partial y}  + 1 \Big )   K_N(x,y) = -  N { \mathcal M[w](N+1)  \over  \mathcal M[w](N)}  p_{N-1}(x) q_N(y),
    \end{equation}
    which implies (\ref{Me3}).
  \end{proposition}
  
  \begin{proof}
  From the formula (\ref{Me1}),
  $$
  x {d \over dx}  p_n(x) = (-1)^n  n!   \mathcal M[w](n+1) \sum_{j=0}^n (-1)^{j} {j \over j! (n-j)!   \mathcal M[w](j+1) } x^j.
 $$
 Rewrite the $j$ in the denominator of this expression as $n - (n - j)$, and use this to decompose the sum into two.
 Upon some simple manipulation, the identity  (\ref{PQ1}) results.
 
 According to (\ref{2.5}), the  formula (\ref{Me2}) can be rewritten
 $$
  q_n(x) = {1 \over n! \mathcal M[w](n+1)}     \prod_{l=1}^{n} \Big ( - x {\partial \over \partial x} - l \Big )  w(x).
$$
 Acting on both sides with $- x {d \over dx} - (n+1)$ shows
 $$
 \Big (  - x {d \over dx} - (n+1) \Big ) q_n(x) = {(n+1) \mathcal M[w](n+2) \over  \mathcal M[w](n+1)} q_{n+1}(x).
 $$ 
This gives (\ref{PQ2}).

With the differential recurrences (\ref{PQ1}) and (\ref{PQ2}) established, we can use them in the expression (\ref{2.9}) to give
\begin{multline}
\Big ( x{\partial \over \partial x} + y{\partial \over \partial y} \Big ) K_N(x,y)  \\
= \sum_{n=0}^{N-1} \Big ( n p_n(x) + n {\mathcal M [w](n+1) \over \mathcal M [w] (n) } p_{n-1}(x) \Big )
\Big (  - {(n+1) \mathcal M[w](n+2) \over  \mathcal M[w](n+1)} q_{n+1}(y) + (n+1) q_n(y) \Big ).
\end{multline}
Simple manipulation reduces this to (\ref{Me4}).

In (\ref{Me4}) scale $x$ and $y$ by writing as $xt$ and $yt$ respectively. The LHS of (\ref{Me4}) can then
be written
\begin{equation}\label{2.17a}
{d \over dt} t K_N(tx,ty) = -  N { \mathcal M[w](N+1)  \over  \mathcal M[w](N)}  p_{N-1}(tx) q_N(ty).
 \end{equation}
Integrating both sides from $0$ to $1$, on the LHS noting $\lim_{t \to 0^+} t K_N(tx,ty) = 0$
as follows from (\ref{2.9}), reclaims  (\ref{Me3}).
  
\end{proof}

\begin{remark}\label{R2}
We show in Appendix B how (\ref{2.17a}), combined with a recurrence formula of fixed depth of
$t p_{N-1}(t)$ known to hold for a number of the specific P\'olya ensembles considered
in Section \ref{S3}, provides a combinatorial based method to compute the leading
large $N$ form of the moments of the spectral density.
\end{remark} 
  
  \subsection{Asymptotics of ratios of gamma function}\label{S2.3}
  
  The gamma function $\Gamma(z)$ is one of the most commonly occurring of special functions \cite{andrews99}, analytic in the
complex plane except for poles at 0 and the negative integers.
Since $\Gamma(z+1) = z \Gamma(z)$ and $\Gamma(1) = 1$, for $n$ a non-negative integer
\begin{equation}\label{0.1} 
 \Gamma(n+1) = n!,
 \end{equation}
 and so gives meaning to the factorial for
general complex $n$. Historically \cite{Pe24} Stirling's formula for the gamma function is the large $n$ approximation to the
factorial $n! \approx \sqrt{2 \pi } n^{n + 1/2} e^{-n}$, later extended to the asymptotic series \cite{WW65}
\begin{equation}
n! =  \sqrt{2 \pi n}  \Big ( {n \over e} \Big )^n \bigg ( 1 + {1 \over 12 n} + {1 \over 288 n^2} + O \Big ( {1 \over n^3} \Big ) \bigg ).
\end{equation}
Using (\ref{0.1}) and
truncating this asymptotic series at $O(1/n)$ leads to the large $|z|$ asymptotic expansion \cite{tricomi51}
\begin{align}\label{eq1}
\frac{\Gamma(z+a)}{\Gamma(z+b)}=z^{a-b}\left(
1+\frac{1}{2z}(a-b)(a+b-1)+O(z^{-2})
\right),\quad |z|\to\infty
\end{align}
valid for $|{\rm arg} \, z|  < \pi$ and $a,b$ fixed.  Furthermore,
specify $(u)_\alpha := \Gamma(u+\alpha)/\Gamma(u)$, which for $\alpha$ a positive integer corresponds to the
product $(u)_\alpha  = (u) (u+1) \cdots (u+ \alpha - 1)$. From this definition, and under the assumption that
$\alpha$ is a positive integer, we see
\begin{align}\label{eq2}
(-N+k)_{\alpha}=   (-1)^\alpha {\Gamma(N - k + 1) \over \Gamma(N - k + 1 - \alpha) } = (-N)^{\alpha}\left(
1-\frac{\alpha(2k+\alpha-1)}{2N}+O(N^{-2})
\right),\quad N\to\infty,
\end{align}
where the large $N$ form follows from (\ref{eq1}).
Our analysis of the rate of convergence for hard edge scalings will have use for both
 (\ref{eq1}) and (\ref{eq2}).
 
 \section{Hard edge scaling to $O(1/N)$ for some P\'olya ensembles}\label{S3}
 \subsection{Products of Laguerre ensembles}\label{S3.1}
The realisation of the Laguerre unitary ensemble with $a = n - N$ noted below (\ref{Lw}) can equivalently be expressed
as  being realised by the squared singular values of an $n \times N$ standard complex Gaussian matrix. A natural
generalisation, first considered in \cite{AKW13,AIK13}, is to consider the squared singular values of the product
of say $M$ rectangular standard complex Gaussian matrices (assumed to be of compatible sizes). Since each ensemble
in the product is individually a P\'olya ensemble, the closure property
of P\'olya ensembles under multiplicative convolution from \cite{KK16} tells us that the product ensemble can be formed by simply replacing $w(x)$
in (\ref{2.3}) by 
\begin{equation}\label{W1}
w^{(M)}(x) := \int_0^\infty dx_1 \cdots dx_M \, \delta  \Big ( x - \prod_{j=1}^M x_j \Big ) \prod_{l=1}^M w_l(x_l), \quad w_j(x) = {1 \over \Gamma(a_j + 1)}  x^{a_j} e^{-x} .
\end{equation}
For the Mellin transform we have the factorised gamma function evaluation
\begin{equation}\label{W2}
\mathcal M[w^{(M)}](s) = \prod_{j=1}^M {\Gamma(a_j + s) \over  \Gamma(a_j + 1)}.
\end{equation}
The formula for the inverse Mellin transform then gives
\begin{align}\label{W3}
w^{(M)}(x)  & =  \Big ( \prod_{j=1}^M {1 \over  \Gamma(a_j + 1)} \Big )   {1 \over 2 \pi i }  \int_{c - i \infty}^{c + i \infty}  \prod_{j=1}^M \Gamma(a_j - s)  \, x^{s} \, ds \nonumber \\
& =   \prod_{j=1}^M {1 \over  \Gamma(a_j + 1)}  \,
\MeijerG{M}{0}{0}{M}{-}{a_1,\ldots,a_M}{x}.
\end{align}
 Here $c$ is any positive real number, and $G^{0,M}_{M,0}$ denotes a particular Meijer G-function; see \cite{MSH09}.
 
 Substituting (\ref{W2}) in (\ref{Me1}) and (\ref{W3}) in (\ref{Me2}) shows \cite{AIK13}
 \begin{align}\label{W4}
 p_n(x) & =   (-1)^n n! \prod_{j=1}^M \Gamma(a_j + n + 1) \sum_{j=0}^n { (-x)^j \over j! (n - j)! \prod_{l=1}^M (a_l +  1)_j}  \nonumber\\
& = 
 (-1)^n \prod_{j=1}^M {\Gamma(a_j + n + 1) \over \Gamma(a_j + 1)} \,
 {}_1 F_M \bigg ( \begin{array}{cc} -n \\ a_1 + 1,\dots, a_M + 1 \end{array} \Big | x \bigg ), 
 \end{align}
 with $ {}_1 F_M$ the notation for the particular hypergeometric series, and
\begin{align}\label{W5}
 q_n(x) & =   {(-1)^n \over n!}  \prod_{j=1}^M {1 \over \Gamma(a_j + n + 1)}  {1 \over 2 \pi i }
   \int_{c - i \infty}^{c + i \infty} {\Gamma(n+s+1) \over \Gamma(s+1) }  \prod_{j=1}^M \Gamma(a_j - s)  \, x^{s} \, ds \nonumber \\
& =  {(-1)^n \over n!}  \prod_{j=1}^M {1 \over  \Gamma(a_j +n+ 1)} \,
\MeijerG{M}{1}{1}{M+1}{-n}{a_1,\ldots,a_M,0}{x}.
 \end{align}
 
 According to (\ref{Me4}) and (\ref{Me3}), $K_N(x,y)$ is fully determined by $p_{N-1}(x)$ and $q_N(y)$. Since our aim is to
 expand $K_N(x,y)$ for large $N$ with hard edge scaled variables, it suffices then to compute the hard edge expansion
 of these particular biorthogonal functions.
 
 \begin{proposition}\label{P3.1}
 Denote
\begin{equation}\label{W6}
{}_0 F_M  \bigg ( \begin{array}{cc} -  \\ a_1 + 1,\dots, a_M + 1 \end{array} \Big | -x \bigg ) = \sum_{j=0}^\infty {(-x)^j \over j! \prod_{s=1}^M (a_s + 1)_j},
\end{equation}
as conforms with standard notation in the theory of hypergeometric functions. We have
\begin{multline}\label{W7}
 {}_1 F_M \bigg ( \begin{array}{cc} -N + 1 \\ a_1 + 1,\dots, a_M + 1 \end{array} \Big | {x \over N} \bigg ) \\ 
 = \bigg ( 1 - {1 \over 2 N} \Big ( x {d \over dx} +
 \Big (  x {d \over dx} \Big )^2 \Big ) \bigg ) \, {}_0 F_M  \bigg ( \begin{array}{cc} -  \\ a_1 + 1,\dots, a_M + 1 \end{array} \Big | -x \bigg ) + O \Big ( {1 \over N^2} \Big ).
 \end{multline}
 Also
 \begin{multline}\label{W8}
 {1 \over N!} \MeijerG{M}{1}{1}{M+1}{-N}{a_1,\ldots,a_M,0}{{x \over N}} \\
 = \bigg ( 1 + {1 \over 2 N} \Big ( x {d \over dx} +
 \Big (  x {d \over dx} \Big )^2 \Big ) \bigg )  \MeijerG{M}{0}{1}{M+1}{-}{a_1,\ldots,a_M,0}{{x }} + O \Big ( {1 \over N^2} \Big ).
 \end{multline} 
 In both (\ref{W7}) and (\ref{W8}) the bound on the remainder holds uniformly for $x \in [0,s]$, for any fixed $s \in \mathbb R_+$.
 \end{proposition}
 
 \begin{proof}
 In the summation (\ref{W4}) defining  the LHS of (\ref{W7}) the only $N$ dependence is the factor
 $$
 {(-N + 1)_j \over N^j} = (-1)^j \bigg ( 1 -\frac{j(j+1)}{2N} +  O \Big ( {1 \over N^2}  \Big )\bigg ),
 $$
 where the expansion follows from (\ref{eq2}). This result, valid for fixed $j$, can nonetheless be substituted in
 the summation since the factor in the summand $(-N+1)_j/j! N^j$ is a rapidly decaying function of $j$. Doing this shows
 \begin{multline*}
  \sum_{j=0}^\infty {(-x)^j \over j! \prod_{s=1}^M (a_s + 1)_j}  \bigg ( 1 - \frac{j(j+1)}{2N}  +  O \Big ( {1 \over N^2}  \Big ) \bigg ) \\ =
\bigg (   1 - {1 \over 2 N} \Big ( x {d \over dx} + \Big ( x {d \over dx} \Big )^2 \Big ) \bigg ) \, {}_0 F_M  \bigg ( \begin{array}{cc} -  \\ a_1 + 1,\dots, a_M + 1 \end{array} \Big | -x \bigg )  +  O \Big ( {1 \over N^2}  \Big ),
\end{multline*}
with the bound on the RHS uniform for $x \in [0,s]$.

In relation to (\ref{W8}), after multiplying through the prefactor $1/ N!$ inside the integrand of the integral (\ref{W5}) defining the LHS,  we see
the only dependence on $N$ is the factor
$$
{\Gamma(N + s +1) \over N^s \Gamma(N+1)} = 1 + {s(s+1) \over 2N}  + O \Big ( {1 \over N^2}  \Big ),
$$
 where the expansion follows from (\ref{eq2}). The result (\ref{W8}) now follows by noting
  \begin{multline*}
 {1 \over 2 \pi i }   \int_{c - i \infty}^{c + i \infty} {1 \over \Gamma(s+1) }  \prod_{j=1}^M \Gamma(a_j - s) \bigg (  1 + {s(s+1) \over 2N}  +  O \Big ( {1 \over N^2}  \Big )  \bigg )  \, x^{s} \, ds  
   \\ = \bigg ( 1 + {1 \over 2 N} \Big ( x {d \over dx} +
 \Big (  x {d \over dx} \Big )^2 \Big ) \bigg )  \MeijerG{M}{0}{1}{M+1}{-}{a_1,\ldots,a_M,0}{x } +  O \Big ( {1 \over N^2}  \Big ),
 \end{multline*}
 and arguing in relation to the error term as above.
\end{proof}  

Substituting the results of Proposition \ref{P3.1} in (\ref{W4}) with $n=N-1$ and in (\ref{W5}) with $n=N$, then substituting in (\ref{Me3}) shows
\begin{multline}\label{FG}
{1 \over N} K_N(x/N,y/N) \\
= \int_0^1 \bigg ( 1 - {1 \over 2 N} \Big ( x {d \over dx} +
 \Big (  x {d \over dx} \Big )^2 \Big ) \bigg ) F(xt)
 \bigg ( 1 + {1 \over 2 N} \Big ( y {d \over dy} +
 \Big (  y {d \over dy} \Big )^2 \Big ) \bigg ) G(yt) \, dt +   O \Big ( {1 \over N^2}  \Big ),
 \end{multline}
 where $F$ denotes the function $ {}_0 F_M$ in (\ref{W7}) and $G$ denotes the function $G_{1,M+1}^{M,0}$ in (\ref{W8}).
 Note that the error bound from asymptotic forms in Proposition \ref{P3.1} persist because the error bounds therein
 are uniform with respect to $x,y$ when these variables are restricted to a compact set; see \cite{Bo16} on this
 point in relation to (\ref{pk5}).
 
 Independent of the details of these functions, the structure (\ref{FG}) permits simplification.
 
 \begin{proposition}\label{P3.2} 
 The expression (\ref{FG}) has the simpler form
 \begin{equation}\label{FG1}
{1 \over N} K_N(x/N,y/N) 
=  \int_0^1  F(xt)  G(yt) \, dt  - {1 \over 2N} \Big ( x {\partial \over \partial x} -  y {\partial \over \partial y} \Big ) F(x) G(y) +   O \Big ( {1 \over N^2}  \Big ).
 \end{equation}
 \end{proposition}
 
 \begin{proof}
 At order $1/N$ the RHS of (\ref{FG}) reads
 $$
 - {1 \over 2 N}   \int_0^1 G(yt)  \Big (  x {d \over dx} + \Big ( x {d \over dx} \Big )^2 \Big ) F(xt) \, dt 
 + {1 \over 2 N}   \int_0^1 F(xt)  \Big (  y {d \over dy} + \Big ( y {d \over dy} \Big )^2 \Big ) G(yt) \, dt.
 $$
 In this expression, both the derivatives with respect to $x$, and the derivatives with respect to $y$ can
 be replaced by derivatives with respect to $t$. Performing one integration by parts for each of
 the terms involving the second derivative, (\ref{FG1}) results.
 \end{proof}
 
 Recalling (\ref{1.5a}), we see from (\ref{FG1}) that in general for  products of Laguerre unitary ensembles, the pointwise rate of
convergence to the hard edge limiting $k$-point correlation is $O(1/N)$. On the other hand, as noted in the text around
(\ref{pk6}), earlier works
\cite{EGP16,Bo16,PS16,HHN16,FT19} have demonstrated that for the Laguerre unitary ensemble itself (the case $M = 1$),
with the hard edge scaling variables as used in (\ref{FG1}), and with the Laguerre parameter $a=0$, the convergence
rate is actually $O(1/N^2)$. Moreover, these same references found that the  $O(1/N^2)$ rate holds for general
Laguerre parameter $a>-1$ if each $N$ on the LHS of (\ref{FG1}) is replaced by $N+a/2$.

From the viewpoint of (\ref{FG1}), the special feature of the case $M=1$ is that then $F$ and $G$ are related
by 
 \begin{equation}\label{FGa}
G(x) = x^a F(x), 
 \end{equation}
as follows from the final paragraph of Section \ref{S2.1}. The term $O(1/N)$ in (\ref{FG1})
can therefore be written to involve only $F$,
 \begin{equation}\label{FG2}
 - {1 \over 2N} y^a   \Big ( -a F(x) F(y) +  \Big ( x {\partial \over \partial x} -  y {\partial \over \partial y} \Big ) F(x) F(y)  \Big )  \Big |_{M=1} .
 \end{equation}
 Substituting in (\ref{FG1}), then substituting the result in (\ref{1.5a}), we factor $x_l$ from each column to
 effectively remove $y^a$ from (\ref{FG2}). The term involving partial derivatives in the latter is then antisymmetric,
 and so does not contribute to an expansion of the determinant at order $1/N$, telling us that
  \begin{align}\label{FG3}
 &{1 \over N^k}  \rho_{(k)}\Big ( {x_1 \over N} ,\dots, {x_k \over N} \Big ) \Big |_{M=1} 
\nonumber  \\
 & \quad = \prod_{l=1}^k x_l^a \det \Big [ \Big ( \int_0^1 t^a F(x_jt) F(x_l t) \, dt + {a \over 2N}F(x_j) F(x_l) \Big )  \Big |_{M=1}  \Big ]_{j,l=1}^k
  + O \Big ( {1 \over N^2}  \Big ) \nonumber \\
 & \quad =  \det \Big [ \Big ( \int_0^1 \tilde{F}(x_jt) \tilde{F}(x_l t) \, dt + {a \over 2N}\tilde{F}(x_j) \tilde{F}(x_l) \Big )  \Big |_{M=1}  \Big ]_{j,l=1}^k
  + O \Big ( {1 \over N^2}  \Big ), 
   \end{align}
   where $\tilde{F}(x) = x^{a/2} F(x)$, and the second equality follows from the first by multiplying each row $j$ by $x_j^{a/2}$ and each column
   $k$ by $x_k^{a/2}$. In this latter form the kernel is symmetric. Comparison with (\ref{1.2}) and (\ref{1.3}) then
   shows
   $$
   \tilde{F}(x)  \Big |_{M=1} = J_a( \sqrt{4x} ), \qquad \int_0^1 \tilde{F}(xt) \tilde{F}(yt) \, dt \Big |_{M=1} = 4 K^{\rm hard}(4x,4y)
   $$
   (the reason for the factors of 4 comes from the choice of hard edge scaling $x \mapsto x/4N$ in  (\ref{1.2}), (\ref{1.3})
   rather than $x \mapsto x/N$ as in (\ref{FG3})). This is in agreement with the references cited above relating
   to the hard edge expansion of the Laguerre unitary ensemble correlation kernel up to and including the
   $O(1/N)$ term, and so has the property
   that upon replacing $N$ by $N+a/2$ on the LHS, the convergence has the optimal rate of $O(1/N^2)$.
   
   \subsection{Laguerre Muttalib-Borodin model}\label{S3.2}
   The Laguerre Muttalib-Borodin model \cite{Mu95,Bo98,FW15,Zh15}, defined as the eigenvalue PDF
   proportional to
   \begin{equation}\label{MB1}
   \prod_{l=1}^N x_l^a e^{-x_l}  \prod_{1 \le j < k \le N} (x_j - x_k ) (x_j^\theta - x_k^\theta),
    \end{equation}
   with each $x_l$ positive is, with $\theta = M$ and upon the change of variables $x_l \mapsto x_l^{1/\theta}$,
   known to be closely related to the product of $M$ matrices from the LUE. Specifically, there is a choice
   of the Laguerre parameters $a_l$ for which the joint PDF of the latter reduces to this transformation of 
   (\ref{MB1}) \cite{KS14}. In particular, it follows that in the case $\theta = M$ at least, (\ref{MB1}) corresponds to a
   P\'olya ensemble.
   In fact it is known from \cite{KK16} that (\ref{MB1}) is an example of a P\'olya ensemble for general $\theta > 0$.
  We can thus make use of the theory of Section \ref{S2.2} to study the hard edge expansion of the correlation
  kernel.
  
  The normalised weight function corresponding to (\ref{MB1}) after the stated  change of variables is
  \begin{equation}\label{MB2} 
  w^{({\rm MB}, L)}(x) = {1 \over \theta \Gamma(a+1)} x^{-1 + (a+1)/\theta} e^{- x^{1/\theta}},
  \end{equation}
  which has Mellin transform
    \begin{equation}\label{MB3} 
    {\mathcal M} [ w^{({\rm MB}, L)} ](s) = {\Gamma(\theta(s-1) + a + 1) \over \Gamma (a + 1) }.
 \end{equation}
 Hence the polynomials $p_n(x)$ in (\ref{Me1}) read
   \begin{equation}\label{MB4}     
   p_n^{({\rm MB}, L)}(x) = (-1)^n \Gamma(\theta n + a + 1) \sum_{j=1}^n { (-n)_j x^j \over j! \Gamma(\theta j + a + 1)},
 \end{equation}
 first identified in the work of Konhauser \cite{Ko67}.
 
 Taking the inverse Mellin transform of (\ref{MB3}) gives the integral form of the weight,
 $$
  w^{({\rm MB}, L)}(x) = {1 \over  \Gamma(a+1)}  {1 \over  2 \pi i} \int_{c - i \theta}^{c + i \theta} \Gamma(-\theta (s + 1) + a + 1) x^s \, ds ,
  $$
  valid for $c > 0$. Using this in (\ref{Me2}) shows
    \begin{equation}\label{MB5} 
    q_n^{({\rm MB}, L)}(x)  = {(-1)^n \over n! \Gamma(\theta n + a + 1)} 
     {1 \over  2 \pi i }   \int_{c - i \theta}^{c + i \theta} { \Gamma(s+n+1) \over \Gamma(s) } \Gamma(-\theta (s + 1) + a + 1) x^s \, ds.
  \end{equation}    
    
 The dependence on $n$ in the summand of (\ref{MB4}) and integrand of (\ref{MB5})  is precisely the same as in 
 (\ref{W4}) and (\ref{W5}) respectively. Applying the working of Proposition \ref{P3.1} then gives hard
 edge asymptotics that is structurally identical to $p_n(x)$ and $q_n(x)$ for products of Laguerre ensembles.
 From this we conclude a formula structurally identical to (\ref{FG1}) for the hard edge asymptotics of the kernel.
 
 \begin{proposition}
 Define
 $$
 \tilde{p}_n^{({\rm MB}, L)}(x) = {(-1)^n \over \Gamma(\theta n + a + 1)} p_n^{({\rm MB}, L)}(x), \qquad
  \tilde{q}_n^{({\rm MB}, L)}(x) = {(-1)^n  \Gamma(\theta n + a + 1)} q_n^{({\rm MB}, L)}(x).
  $$
  Also define
  $$
  F^{({\rm MB},L)}(x) = \sum_{j=0}^\infty {x^j \over j! \Gamma(\theta j + a + 1)}, \qquad
 G^{({\rm MB},L)}(x) =   {1 \over 2 \pi i}    \int_{c - i \theta}^{c + i \theta} { \Gamma(-\theta (s + 1) + a + 1) \over \Gamma (s) } x^s \, ds.
 $$
 We have
 \begin{align*}
 \tilde{p}_{N-1}^{({\rm MB}, L)}(x/N) & = \bigg ( 1 - {1 \over 2 N} \Big ( x {d \over dx} + \Big (  x {d \over dx} \Big )^2 \Big )   \bigg ) F^{({\rm MB},L)}(x)  + O \Big ( {1 \over N^2} \Big ) \\
  \tilde{q}_N^{({\rm MB}, L)}(x/N) & = \bigg ( 1 + {1 \over 2 N} \Big ( x {d \over dx} + \Big (  x {d \over dx} \Big )^2 \Big )  \bigg ) G^{({\rm MB},L)}(x)  + O \Big ( {1 \over N^2} \Big ) ,
  \end{align*}
  and furthermore
  \begin{multline*}
  {1 \over N} K_N^{({\rm MB}, L)}(x/N,y/N) \\
=  \int_0^1  F^{({\rm MB}, L)}(xt)  G^{({\rm MB}, L)}(yt) \, dt  - {1 \over 2N} \Big ( x {\partial \over \partial x} -  y {\partial \over \partial y} \Big ) F^{({\rm MB}, L)}(x) G^{({\rm MB}, L)}(y) +   O \Big ( {1 \over N^2}  \Big ).
\end{multline*}
  \end{proposition}
  
  As in the discussion following Proposition \ref{P3.1}, this tells us that the rate of convergence to the hard edge scaled limit of the
  $k$-point correlation is $O(1/N)$, with the case $\theta = 1$ (corresponding to the LUE) an exception, where by appropriate choice
  of scaling variables, the rate is $O(1/N^2)$.
  
  \subsection{Products of Laguerre ensembles and inverse Laguerre ensembles}\label{S3.3}
  In the guise of the square singular values for the product of complex Gaussian matrices, times the inverse of a further
  product of complex Gaussian matrices, the study of the eigenvalues of a product of Laguerre ensembles and  inverses
  was initiated in \cite{Fo14}. This was put in the context of P\'olya ensembles in \cite{KS14}. Moreover, in the case that there
  are equal numbers of matrices and inverse matrices, such product ensembles can be related to a single weight function,
  as we will now demonstrate.  The essential point
  is that the eigenvalues of $X_{b_1}^{-1} X_{a_1}$, where $X_{a_1}$, $X_{b_1}$  has eigenvalues from the Laguerre unitary ensemble
  has eigenvalue PDF proportional to (see e.g.~\cite[Exercises 3.6 q.3]{Fo10})
  \begin{equation}\label{Kq1}
  \prod_{l=1}^N {x_l^{a_1} \over (1+ x_l)^{b_1 + a_1 + 2N}} \prod_{1 \le j < k \le N} ( x_k - x _j)^2
  \end{equation}
  and that this in turn is an example of a  P\'olya ensemble (\ref{2.3}) with
   \begin{equation}\label{Kq2}
   w^{(\rm I)}(x) = {x^{a_1} \over  (1 + x)^{b_1 + a_1+N+1}} \chi_{x > 0}
    \end{equation} 
    (here the superscript (I) indicates `inverse').
  Structurally, a key distinguishing feature relative to the weight (\ref{Lw}) is that (\ref{Kq2}) depends on $N$.
  After normalising (\ref{Kq2}), proceeding as in the derivation of (\ref{W1}) shows the weight function
  for the P\'olya ensemble of the corresponding product ensemble is
   \begin{equation}\label{Kq3}
   \mathcal M [   w^{({\rm I}, M)}](s)  =  \prod_{l=1}^M {  \Gamma(a_l+s) \Gamma(b_l + N  + 1- s) \over \Gamma(a_l+1) \Gamma(b_l + N) }.
      \end{equation} 
   Use of (\ref{Kq3}) in  (\ref{Me1}) shows
   \begin{multline}\label{Kq4}
{(-1)^n  \over      \prod_{l=1}^M    \Gamma(a_l+n+1) \Gamma(b_l + N  - n ) } 
 p_n^{({\rm I}, M)}(x) 
 = \sum_{j=0}^n  {(-n)_j \over j!}   {x^j \over  \prod_{l=1}^M    \Gamma(a_l+j+1) \Gamma(b_l + N  - j )  } . 
 \end{multline}     
 Further, using  (\ref{Kq3}) to write $w^{(\rm I)}(x) $ as an inverse Mellin transform shows from (\ref{Me2}) that
    \begin{multline}\label{Kq5}
  {(-1)^n    \over  \prod_{l=1}^M    \Gamma(a_l+n+1) \Gamma(b_l + N  - n ) }q_n^{({\rm I}, M)}(x)  \\ = {1 \over 2 \pi i}  {1 \over n!} \int_{c - i \theta}^{c + i \theta} 
   { \Gamma(s+n) \over \Gamma(s) }  \Big (  \prod_{l=1}^M
   \Gamma(a_l-s) \Gamma(b_l + N  +1 + s)  \Big ) x^s \, ds.
 \end{multline}       
  
  Proceeding as in the derivation of Proposition \ref{P3.1}, and making use in particular of the asymptotic formula
  (\ref{eq1}) for the ratio of two gamma functions, the large $N$ forms of (\ref{Kq4}) and (\ref{Kq5}) as relevant to (\ref{Me3})
 can be deduced. This allows for the analogue of (\ref{FG}) to be deduced, which then proceeding as in the derivation of
 Proposition \ref{P3.2} gives the analogue of (\ref{FG1}).
 
 \begin{proposition}\label{P3.4}
 Denote the LHS of (\ref{Kq4}) with $n = N -1$, and multiplied by $\prod_{l=1}^M \Gamma( N + b_l)$, by
 $\tilde{p}_{N-1}^{({\rm I}, M) }(x)$,  and let $F$ be specified as below (\ref{FG}). Also, denote the LHS of
(\ref{Kq5}) with $n = N $, and divided by $\prod_{l=1}^M \Gamma( N + b_l)$, by  $\tilde{q}_{N}^{({\rm I}, M) }(x)$,
and let $G$ be as specified below (\ref{FG}). We have
\begin{equation}
p_{N-1}^{({\rm I}, M)}\Big ( {x \over N^{M+1}} \Big ) 
= \bigg ( 1   - {1 \over 2 N} \bigg (
\Big ( 1 + M - 2 \sum_{l=1}^M b_l  \Big ) x {d \over dx} +
(1 + M) \Big ( x {d \over d x} \Big )^2  + O\Big ( {1 \over N^2} \Big ) \bigg ) \bigg ) F(x),
\end{equation}
\begin{multline}
 {1 \over N^M}  q_{N}^{({\rm I}, M)}\Big ( {x \over N^{M+1}} \Big ) \\
= \bigg ( 1 +  {1 \over N} \sum_{l=1}^M  b_l  +{1 \over 2N}  \bigg (
\Big ( 1 + M + 2 \sum_{l=1}^M b_l  \Big ) x {d \over dx} +
(1 + M) \Big ( x {d \over d x} \Big )^2  + O\Big ( {1 \over N^2} \Big )  \bigg ) \bigg ) G(x)
\end{multline}
and
\begin{multline}\label{3.26}
{1 \over  N^{M+1}}  K_N \Big ( {x \over N^{M+1}}, {y \over N^{M+1}} \Big ) = \int_0^1 F(xt) G(yt) \, dt \\
 - {1 \over 2 N} (1 + M) \Big ( G(y) x {d \over dx} F(x) - F(x) y {d \over d y} G(y) \Big )
+ {1 \over  N} \Big ( \sum_{l=1}^M b_l  \Big ) F(x) G(y) + O \Big ( {1 \over N^2} \Big ).
\end{multline}
\end{proposition}

The expansion (\ref{3.26}) shows that in general the leading correction to the 
hard edge scaled limit of the $k$-point correlation
in the case of $M$ products of random matrices formed from the multiplication of a
Laguerre unitary ensemble and inverse Laguerre unitary ensemble is $O(1/N)$.
However, as for products studied in Section \ref{S3.1}, the case $M=1$ is special,
as then the relation (\ref{FGa}) between $F$ and $G$ holds. The $O(1/N)$ term in
(\ref{3.26}) the simplifies to read
\begin{equation}\label{3.27}
{1 \over N} y^a \bigg ( (a_1 + b_1)  F(x) F(y)  - \Big ( x {\partial \over \partial x} - y {\partial \over \partial y} \Big ) F(x) F(y) \bigg ) \bigg |_{M=1}
\end{equation}
Proceeding now as in the derivation of (\ref{FG3}), and with the same meaning of $\tilde{F}$ used
therein, we thus have
 \begin{align}\label{FG3a}
 &{1 \over N^{2k}}  \rho_{(k)}\Big ( {x_1 \over N^2} ,\dots, {x_k \over N^2} \Big ) \Big |_{M=1} 
\nonumber  \\
 & \quad = \prod_{l=1}^k x_l^{a_1} \det \Big [ \Big ( \int_0^1 t^{a_1} F(x_jt) F(x_l t) \, dt + {a_1 + b_1  \over N}F(x_j) F(x_l) \Big )  \Big |_{M=1}  \Big ]_{j,l=1}^k
  + O \Big ( {1 \over N^2}  \Big ) \nonumber \\
 & \quad =  \det \Big [ \Big ( \int_0^1 \tilde{F}(x_jt) \tilde{F}(x_l t) \, dt + {a_1 + b_1  \over N}\tilde{F}(x_j) \tilde{F}(x_l) \Big )  \Big |_{M=1}  \Big ]_{j,l=1}^k
  + O \Big ( {1 \over N^2}  \Big ).
   \end{align}
  As in the discussion below (\ref{FG3}), it follows that if on the LHS $N$ is replaced by $N + (a_1 + b_1)/2$, the convergence to
  the hard edge limit has the optimal rate of $O(1/N^2)$.
  
  \begin{remark}
  1.~Changing variables $x_l = y_l/(1 - y_l)$, $0 < y_l < 1$ in (\ref{Kq1}) gives the functional form
  \begin{equation}\label{Ja}
  \prod_{l=1}^N y_l^{a_1} ( 1 - y_l)^{b_1} \prod_{1 \le j < k \le N} (y_k - y_j)^2,
  \end{equation}
  which up to proportionality is the eigenvalue PDF for the Jacobi unitary ensemble. In the recent work
  \cite{MMM19} the corrections to the hard edge scaled limit of the distribution of the smallest
  eigenvalue have been analysed, with results obtained consistent with (\ref{FGa}). In Appendix
  A we present a large $N$ analysis of this distribution for the Jacobi $\beta$-ensemble (the Jacobi unitary ensemble
  is the case $\beta = 2$) for general even $\beta$. \\
  2.~The case $b_1 = 0$ of the Jacobi unitary ensemble is closely related to the Cauchy two-matrix
  model \cite{BGS14}. The latter is determinantal, but since the PDF consists of two-components, 
  the determinant has a block structure. Nonetheless, each block can be expressed in terms of just
  a single correlation kernel. The hard edge scaling of the latter has been undertaken in \cite{BGS14},
  with a result analogous to (\ref{FG3a}) with $b_1 = 0$ obtained. Closely related to the Cauchy two-matrix
  matrix model is the Bures ensemble, as first observed in \cite{BGS09}, and further developed
  in \cite{FK16}, with a Muttalib-Borodin type extension given in \cite{FL19}. Since the elements of the
  correlation kernel for the Bures ensemble (which is a Pfaffian point process)  are given in terms of
  the correlation kernel for the Cauchy two-matrix matrix model, it follows that by tuning the scaling
  variables at the hard edge, an optimal convergence rate of $O(1/N^2)$ can be achieved.
   \\
    3.~A Muttalib-Borodin type generalisation of (\ref{Kq1}) is known \cite[Jacobi prime case]{FI18}. Working
  analogous to that of Section \ref{S3.2} could be undertaken, although we refrain from doing that here.
  It would similarly be possible to obtain the analogue of Proposition \ref{P3.4} for the singular values
  of products of truncations of unitary ensembles \cite{KKS15}, which we know from \cite{KK16} can
  be cast in a P\'olya ensemble framework as products of Jacobi unitary ensembles. 
  \end{remark}

   \subsection*{Acknowledgements}
	This research is part of the program of study supported
	by the Australian Research Council Centre of Excellence ACEMS. We thank Mario Kieburg for
	feedback on a draft of this work.

  \appendix
\section*{Appendix A}
\renewcommand{\thesection}{A} 
\setcounter{equation}{0}
  
In random matrix theory there is special importance associated with the $\beta$ generalisation of (\ref{1.1}), specified
by the class of PDFs proportional to
\begin{equation}\label{A.1}
\prod_{l=1}^N w(x_l) \prod_{1 \le j < k \le N} |x_k - x_j|^\beta.
 \end{equation}
 The parameter $\beta$ is referred to as the Dyson index \cite{Dy62}, and in classical random matrix theory corresponds
 to the matrix ensemble being invariant with respect to conjugation by real orthogonal ($\beta = 1$),
 complex unitary ($\beta = 2$) and unitary symplectic matrices ($\beta = 4$). For general $\beta > 0$, (\ref{A.1})
 has the interpretation as the Boltzmann factor of a classical statistical mechanical system with particles
 repelling via the pair potential $- \log | x - y|$, confined by a one-body potential with Boltzmann factor $w(x)$,
 and interacting at the inverse temperature $\beta$. Also, with $w(x)$ one of the classical weights --- Gaussian,
 Laguerre or Jacobi --- (\ref{A.1}) for general $\beta > 0$ is the exact ground state wave function for particular
 quantum many body systems of Calogero-Sutherland type (this requires a change of variables in the Laguerre
 and Jacobi cases; see \cite{BF97a}).
 
 Our interest is in (\ref{A.1}) with the Jacobi weight (\ref{A.2}).
 Details of various realisations of  (\ref{A.1}) as an eigenvalue PDF in this case can be found in \cite[\S 1.1]{FK20}.
 While there are no tractable formulas for the $k$-point correlation functions for general $\beta > 0$, it turns out that
 for a particular class of Jacobi gap probabilities $E_{N,\beta}(0;J;w(x))$ --- this denoting the probability that there are no eigenvalues in the
 interval $J$ for the ensemble specified by the eigenvalue PDF (\ref{A.1}) --- evaluations are available in terms of particular
 multivariate hypergeometric functions; see \cite[Ch.~12 \& 13]{Fo10}, which are well suited to the analysis of the rate
 of convergence to the hard edge limit. This circumstance similarly holds true for the Laguerre case of (\ref{A.1}), for which
 an analysis of the rate of convergence has recently been carried out in \cite{FT19}.
 
 The starting point is the fact that for $J = (s,1)$, and for the parameter $b \in \mathbb Z_{\ge 0}$, a simple change
 of variables in the multi-dimensional integral defining $E_{N,\beta}(0;J;w(x))$ shows that as function of $s$ it
 is a power function times a polynomial (see \cite[\S 1.3]{FK20} for details),
  \begin{equation}\label{A.3}
E_{N,\beta}(0;(s,1);x^a(1-x)^b)=s^{N(a+1)+\beta N(N-1)/2}\sum_{p=0}^{bN}\gamma_p s^p,
\end{equation}
for some  coefficients $\gamma_p$. Moreover, we know from \cite[Eq.~(13.7) and Prop.~13.1.7]{Fo10} that this
polynomial can be identified as a particular multivariate hypergeometric function, generalising the Gauss hypergeometric
function
  \begin{equation}\label{A.4}
E_N(0;(s,1);x^a(1-x)^b)=s^{N(a+1)+\beta N(N-1)/2}  {}_2^{}F_1^{(\beta/2)} (
-N,-(N-1)-{2(a+1)}/\beta; 2b/\beta;(1-s)^b ).
\end{equation}
 In the last argument, the notation $(1-s)^b$ refers to $1-s$ repeated $b$ times. In the case $b = 1$,
 ${}_2^{}F_1^{(\beta/2)}$ coincides with the Gauss hypergeometric function independent of $\beta$.

 For general positive integer $b$ we will make use of the $b$-dimensional integral representation  \cite[Eq. (13.11)]{Fo10}
   \begin{multline}\label{A.5}
{}_2^{}   F_1^{(\beta/2)} (
r,-\tilde{b},\frac{2(b-1)}{\beta}+\tilde{a}+1;(u)^b ) =  \frac{1}{M_b(\tilde{a},\tilde{b},2/\beta)}\\
\times \int_{-1/2}^{1/2}dx_1\cdots\int_{-1/2}^{1/2}dx_b \, \prod_{l=1}^b e^{\pi ix_l (\ta-\tb)}|1+e^{2\pi i x_l}|^{\ta+\tb}(1+ue^{2\pi ix_l})^{-r} \prod_{1\leq j<k\leq b}|e^{2\pi i x_k}-e^{2\pi i x_j}|^{4/\beta} \\
=
\frac{N^{b\ta}}{M_b(\ta,\tb,2/\beta)} 
\int_{\mathcal C^b}  dx_1 \cdots dx_b \,  \prod_{l=1}^b e^{2\pi i x_l\ta}(1+N^{-1}e^{-2\pi i x_l})^{\ta+\tb}(1+uNe^{2\pi i x_l})^{-r} \\ \times
 \prod_{1\leq j<k\leq b}|e^{2\pi i x_k}-e^{2\pi i x_j}|^{4/\beta}
\end{multline}
for the parameters $r=-N$, $\tilde{b}=(N-1)+(2/\beta)(a+1)$, $\tilde{a}=2/\beta-1$. Here
 the normalisation $M_b(\ta,\tb,2/\beta)$ is the Morris integral, with gamma function evaluation
 (see e.g.~\cite[Eq. (1.18)]{FW08})
\begin{equation}\label{Mb}
M_b(\ta,\tb,2/\beta)=\prod_{j=0}^{b-1}\frac{\Gamma(1+\ta+\tb+2j/\beta)\Gamma(1+2(j+1)/\beta)}{\Gamma(1+\ta+2j/\beta)
\Gamma(1+\tb+2j/\beta)\Gamma(1+2/\beta)}.
\end{equation}
The second equality follows by
manipulating the integrand so that it is an analytic function of $z_l=e^{2\pi i x_l}$, then changing variables $z_l\mapsto z_lN$, and finally deforming each circle contour
to a contour $\mathcal C_z$, as detailed in \cite[Prop.~2]{Fo13}, and to be described next. It
starting at the origin in the complex $z$-plane, running along the negative real axis in the bottom
half plane to $z = - 1 - 0i$, then along a counter clockwise circle to $z = - 1 + 0i$, and finally back to the origin along the negative real axis in the upper
half plane.
 The contour $\mathcal C$ is the image of $\mathcal C_z$ in the complex $x$-plane under the mapping $z = e^{2 \pi i x}$.
 With an appropriate scaling of $u$, this second multidimensional integral is well suited to an asymptotic analysis, enabling an asymptotic analysis of
the hard edge limit in (\ref{A.4}).

To identify a structured form in the resulting expression, we have need for knowledge of inter-relations satisfied by
the multiple integrals
\begin{equation}\label{Mb1}
I_b(s)[f]:= \int_{\mathcal C^b} dx_1 \cdots dx_b \, f(x_1,\dots,x_b)   \prod_{l=1}^b e^{2\pi i x_l(2/\beta-1)}e^{e^{-2\pi i x_l}+(s/4)e^{2\pi i x_l}}
   \prod_{1\leq j<k\leq b}|e^{2\pi i x_k}-e^{2\pi i x_j}|^{4/\beta} 
   \end{equation}
   for $f = f_q := \sum_{l=1}^b e^{2 \pi i q x_l}$, $q=0,\pm 1, \pm 2$. The simplest, which follows immediately from the definitions, is that
 \begin{equation}\label{Mb2}
 {1 \over b} {d \over d s}   I_b(s)[f_0] = {1 \over 4}   I_b(s)[f_1] .
 \end{equation}
 Integration by parts techniques, well known in the theory of the Selberg integral \cite{Ao87}, \cite[\S 4.6]{Fo10}, reveals further relations.

\begin{proposition}\label{PA1}
We have
\begin{align*}
{s \over 16} I_b(s)[f_2] & = - {2 \over \beta} {d \over ds} I_b(s)[f_0] + {1 \over 4}  I_b(s)[f_0] \\
I_b(s)[f_{-2}] & = {s \over 4}  I_b(s)[f_0]  + 2 \Big (  {2 \over \beta} - 1  - {b \over \beta} \Big ) \Big ( {s \over b}  {d \over ds}  I_b(s)[f_0]  +
 \Big ( {2 \over \beta} - 1 \Big )  I_b(s)[f_0]  \Big ) \\
 I_b(s)[f_{-1}] & =  \Big ( {2 \over \beta} - 1 \Big )  I_b(s)[f_0]  + {s \over b}  I_b(s)[f_0].
\end{align*}
 \end{proposition} 
 
 \begin{proof}
 According to the fundamental theorem of calculus
 $$
 I_b(s) \Big [ \sum_{l=1}^b {\partial \over \partial x_l} e^{2 \pi i x_l} \Big ] = 0.
 $$
 Performing the differentiations on the LHS, this implies
\begin{multline*}
0 = \frac{2}{\beta}  I_b(s)[f_{1}]  -  I_b(s)[f_{0}] 
+\frac{s}{4}  I_b(s)[f_{2}]  \\ +\frac{2}{\beta} I_b(s) \bigg [ \sum_{l\ne k}^b 
e^{2\pi i x_l}\left(
\frac{e^{2\pi i x_l}}{e^{2\pi i x_l}-e^{2\pi i x_k}}+\frac{e^{-2\pi i x_l}}{e^{2\pi i x_l}-e^{2\pi i x_k}}
\right)\bigg ] = 0.
\end{multline*}
Symmetrising the integrand in the final average reduces this to
\begin{align*}
\frac{2b}{\beta}  I_b(s)[f_1]  -  I_b(s)[f_{0}]  
+\frac{s}{4}  I_b(s)[f_{2}] = 0.
\end{align*}
Recalling now (\ref{Mb2}) gives the first of the stated relations.

The other two follow by similar working. In fact they have been derived
previously; see \cite[\S 3.2]{FT19}.
 \end{proof} 

\begin{proposition}
Define
  \begin{multline}
  E^{\rm hard}(s;b) =
 \frac{e^{-\beta s/8}b!}{(\Gamma(2/\beta))^b}   \\
  \times
 \int_{\mathcal C^b} 
 dx_1\cdots dx_b  \prod_{l=1}^b e^{2\pi i x_l(2/\beta-1)}e^{e^{-2\pi i x_l}+(s/4)e^{2\pi i x_l}}  \prod_{1\leq j<k\leq b}|e^{2\pi i x_k}-e^{2\pi i x_j}|^{4/\beta}.
\end{multline}
For general $\beta>0$ and $b\in \mathbb{Z}_{\geq0}$, we have
\begin{equation}\label{St0}
E_N(0;(1-s/4N^2,1);x^a(1-x)^b)=  E^{\rm hard}(s;b) +\frac{1}{N}\left(
\frac{2(1+a+b)}{\beta}-1
\right)s {d \over d s}  E^{\rm hard}(s;b)  +O \Big ( {1 \over N^2} \Big ).
\end{equation}

\end{proposition} 

\begin{proof}
According to (\ref{A.4}), the analysis of $E_N(0;(1-s/4N^2,1);x^a(1-x)^b)$ requires replacing $u$ by $s/4N^2$ in (\ref{A.5}).
With this done, we see there is a dependency on $N$ both outside and inside the integral. For both, the large $N$ form
can readily be computed. The factor outside the integral involves the Morris integral, which has the evaluation (\ref{Mb}).
Recalling the values of  $\ta$ and $\tb$, and use  of the ratio of gamma function asymptotic formula \eqref{eq1}
shows
\begin{align*}
\frac{N^{b\ta}}{M_b(\ta,\tb,2/\beta)}=\frac{(\Gamma(2/\beta))^b}{b!}\left(
1-\frac{(2/\beta-1)^b}{N}\Big(
\frac{2a+b+1}{\beta}-1 \Big ) +O\Big ( {1 \over N^2}  \Big )
\right).
\end{align*}
For the $N$ dependent factors in the integrand, a simple power series expansion shows
\begin{multline*}
\prod_{l=1}^b e^{2\pi i x_l\ta}(1+N^{-1}e^{-2\pi i x_l})^{\ta+\tb}(1+(s/4N)e^{2\pi i x_l})^{-r}=\left(
\prod_{l=1}^b e^{2\pi i x_l\ta}e^{e^{-2\pi i x_l}+(s/4)e^{2\pi i x_l}}\right) \\
 \times\left(1+\frac{1}{N}\left(
-2+\frac{2}{\beta}(a+2)
\right)\sum_{l=1}^b e^{-2\pi i x_l}-\frac{1}{2N}\sum_{l=1}^b e^{-4\pi i x_l}-\frac{s^2}{32N}\sum_{l=1}^be^{4\pi i x_l}+ O\Big ( {1 \over N^2}  \Big )
\right).
\end{multline*}

Substituting these expansions in (\ref{A.5}), we see from (\ref{A.4}) that
\begin{multline}\label{St1}
E_N(0;(1-s/4N^2,1);x^a(1-x)^b) 
 =\frac{e^{-\beta s/8}b!}{(\Gamma(2/\beta))^b} \\
 \times \int_{\mathcal C^b} dx_1 \cdots \ dx_b \,
  \prod_{l=1}^b e^{2\pi i x_l(2/\beta-1)}e^{e^{-2\pi i x_l}+(s/4)e^{2\pi i x_l}}   \prod_{1\leq j<k\leq b}|e^{2\pi i x_k}-e^{2\pi i x_j}|^{4/\beta}  \\
\times\left\{
1+\frac{1}{N}\left[\frac{s\beta}{8}
\left(1-\frac{2(a+1)}{\beta}\right)-\left(
\frac{2}{\beta}-1
\right)b\left(
\frac{2a+b+1}{\beta}-1
\right)\right.\right. 
\\ \left.\left.+\left(-2+\frac{2}{\beta}(a+2)\right)\sum_{l=1}^b e^{-2\pi i x_l}-\frac{1}{2}\sum_{l=1}^b e^{-4\pi i x_l}-\frac{s^2}{16}\sum_{l=1}^b e^{4\pi i x_l}
\right]+ O\Big ( {1 \over N^2}  \Big ) \right\}.
\end{multline}
At $O(1/N)$ the multidimensional integral in this expression can be written in terms of the notation (\ref{Mb1}) as
\begin{multline*}
\left[\frac{s\beta}{8}
\left(1-\frac{2(a+1)}{\beta}\right)-\left(
\frac{2}{\beta}-1
\right)b\left(
\frac{2a+b+1}{\beta}-1
\right) \right ] {1 \over b} I_b[s][f_0] \\
+ \left(-2+\frac{2}{\beta}(a+2)\right)  I_b[s][f_{-1}] - \frac{1}{2}  I_b[s][f_{-2}] -\frac{s^2}{16}\ I_b[s][f_{2}] .
\end{multline*}
After simplification using Proposition \ref{PA1}, and substitution back in (\ref{St1}), an expansion
equivalent to (\ref{St0}) results.

\end{proof}

  \appendix
\section*{Appendix B}
\renewcommand{\thesection}{B} 
\setcounter{equation}{0}
The application given to (\ref{2.17a}) in the main text is to derive the integral form of the kernel
(\ref{Me3}). Another application relates to the moments of the spectral density, since setting
$x=y=1$, multiplying both sides by $t^p$, and integrating both sides from $0$ to $\infty$
using integration by parts on the LHS shows
\begin{equation}\label{B.1}
k \int_0^\infty t^k K_N(t,t)  \, dt=     N { \mathcal M[w](N+1)  \over  \mathcal M[w](N)}  \int_0^\infty t^k p_{N-1}(t) q_N(t) \, dt.
\end{equation}
And since the P\'olya ensembles are determinantal, $K_N(t,t) = \rho_{(1)}(t)$, so the LHS is $k$ times
the $k$-th moment of the spectral density. 

Suppose now for some fixed $r \in \mathbb Z^+$, and any fixed $i \in \mathbb Z$
\begin{equation}\label{B.2}
t  p_{N-i}(t) =  \sum_{s=-r}^1 \alpha_{N-i,s} p_{N-i+s}(t).
\end{equation}
Moreover, suppose that the coefficients $\alpha_{N-1,s}$ have the large
$N$ form $\alpha_{N-i,s}/N^{\hat{r}} \to \hat{\alpha}_{s}$ for some
$\hat{r}$, and so
\begin{equation}\label{B.2a}
t  p_{N-i}(t) \mathop{\sim}\limits_{N \to \infty} N^{\hat{r}}  \sum_{s=-r}^1 \hat{\alpha_{s}} p_{N-i+s}(t).
\end{equation}
We begin by substituting for $t p_{N-1}(t)$ in (\ref{B.1}) using (\ref{B.2a}) with $i=1$. 
In the case $k=1$ only the term $s=1$ contributes due to the orthogonality (\ref{2.8}), so the integral in (\ref{B.1}) has
the large $N$ evaluation $N^{\hat{r}} \hat{\alpha}_1$.

For $k\ge 2$ we next use (\ref{B.2a}) to expand $t p_{N-i+s}(t)$, and in so doing
reducing the exponent in the integrand down to $k-2$. In the case $k=2$ the orthogonality (\ref{2.8})
implies  the integral in (\ref{B.1}) has the large $N$ evaluation $2 N^{2 \hat{r}} \hat{\alpha}_0  \hat{\alpha}_1$.
For $k \ge 3$ we continue by a further use (\ref{B.2a}), reducing the power in (\ref{B.1}) down to $t^{k-3}$,
and repeat so after a total of $k$ applications of  (\ref{B.2a}) the integrand is a linear combination of
$\{ p_l(t) \}$ times $q_N(t)$. By the orthogonality (\ref{2.8}), only the coefficient of $p_N(t)$ in the linear combination contributes to the
integral in (\ref{B.1}). Each term in the linear combination can be related to a weighted lattice path, consisting of
$k$ steps, which at each step and for some $s=1,0,\dots,-r$ changes height by $s$ units. Only those paths which
change height by a total of exactly one unit make up the coefficient of $p_N(t)$, showing that
\begin{equation}\label{B.3}
k \lim_{N \to \infty} {1 \over  N^{k \hat{r}+1} }  { \mathcal M[w](N) \over  \mathcal M[w](N+1)}  \int_0^\infty t^k K_N(t,t)  \, dt =   
\sum_{R}   \bigg ( {k \atop a_1,a_0,\dots,a_{-r}} \bigg )  \prod_{s=-r}^1  \hat{\alpha}_{s}^{a_s},
\end{equation}
where the restriction $R$ on the non-negative integers $a_1,\dots, a_{-r}$ is specified by
\begin{equation}\label{B.4}
 R: \quad    \sum_{s=-r}^1  a_s = k, \: \:  \sum_{s=-r}^1 s a_s = 1 
\end{equation}	
(cf.~\cite[Prop.~2.6]{Ha18}). Furthermore, we observe that with $[u]f(u)$ denoting the coefficient of $u$ in
the power series expansion of $f(u)$ the sum in (\ref{B.3}) can be expressed in terms of a generating function
according to
\begin{equation}\label{B.4a}
\sum_{R}   \bigg ( {k \atop a_1,a_0,\dots,a_{-r}} \bigg )  \prod_{s=-r}^1  \hat{\alpha}_{s}^{a_s} =
[u] \, \Big ( u \hat{\alpha_1} +  \hat{\alpha_0} + \cdots + u^{-r}  \hat{\alpha_r} \Big )^k.
\end{equation}	

Let us specialise now to the product of $M$ Laguerre ensembles as in Section \ref{S3.1}.
For convenience, with $p_n(x)$ given by (\ref{W4}), introduce the rescaled polynomial
\begin{equation}\label{B.5}
P_n(x) = {1 \over c_n } p_n(x), \qquad   c_n = n!   \mathcal M[w](n+1).
\end{equation}
The advantage of this normalisation is that the recurrences corresponding to (\ref{B.2}) and its large
$N$ asymptotics (\ref{B.2a}) have been computed by Lambert \cite[Props.~4.3 \& 4.10]{La18}, with
the latter reading
\begin{equation}\label{B.5a}
t P_{N-i}(t)  \mathop{\sim}\limits_{N \to \infty}  N^M \sum_{s=-M}^1 N^s \hat{\beta}_s P_{N-i+s}(t), \qquad  \hat{\beta}_s = \binom{M+1}{-s+1}.
\end{equation}
Proceeding as in the derivation of (\ref{B.3}), and making use too of (\ref{B.4a}),  we see that for $k \ge 1$,
\begin{multline}\label{B.6}
k \lim_{N \to \infty} {1 \over  N^{k M+1} }  \int_0^\infty t^k K_N(t,t)  \, dt  =   [u] \, \Big ( u \hat{\beta_1} +  \hat{\beta_0} + \cdots + u^{-M}  \hat{\beta_M} \Big )^k 
\\ = [u^{k+1}]  ( 1 + 1/u)^{k(M+1)}
 = \binom{k(M+1)}{k+1},
\end{multline}
where the second equality follows by recognising the series, with the $\hat{\beta}_s$ as in (\ref{B.5a}), as a binomial expansion, so it can be summed, while the third equality
follows by applying the binomial expansion to power series expand the resulting expression.
Here we recognise
\begin{equation}\label{B.7}
{1 \over k}  \binom{k(M+1)}{k+1} =   {1 \over kM + 1}  \binom{k(M+1)}{k} 
\end{equation}
as the $k$-th Fuss-Catalan number, indexed by $M$, with the Catalan numbers the case $M=1$.
This combinatorial sequence is well known
to give the scaled moments of the spectral density for the product of $M$ Laguerre ensembles (or equivalently the
scaled moments of the squared singular values of the product of $M$ standard complex Gaussian matrices);
see \cite{PZ11,FL15}.

%
%

  \small
\bibliographystyle{abbrv}

\end{document}